\documentclass[letterpaper,twocolumn,10pt]{article}
\usepackage{usenix}

\usepackage[utf8]{inputenc}
\usepackage[colorlinks=true,allcolors=blue,breaklinks,draft=false]{hyperref}

\usepackage{amsmath,amsthm,stmaryrd,amsfonts,mathtools}
\usepackage{algorithm,algpseudocode}
\usepackage{subcaption}
\usepackage{booktabs}
\usepackage{appendix}
\usepackage{enumitem}

\algrenewcommand\algorithmicindent{1em}

\newcommand*\circled[1]{\tikz[baseline=(char.base)]{
            \node[shape=circle,draw,inner sep=0.5pt] (char) {#1};}}

\usepackage[all]{xy}
\usepackage{cleveref}
\crefname{section}{\S}{\S\S}
\Crefname{section}{\S}{\S\S}

\usepackage{etex}
\usepackage{tikz}
\usepackage{pgfplots}
\usetikzlibrary{shapes, calc, backgrounds, pgfplots.statistics, arrows,positioning, calc}
\pgfplotsset{compat=1.8}
\tikzstyle{vertex}=[]

\usepackage{gnuplottex}

\usepackage{verbatim}

\DeclarePairedDelimiter\abs{\lvert}{\rvert}%

\newcommand{\bounds}{\mathfrak{M}}

\newcommand{\atom}[1]{\alpha_{#1}}
\newcommand{\rlower}[1]{\mathsf{lower}(#1)}
\newcommand{\rupper}[1]{\mathsf{upper}(#1)}
\newcommand{\rpriority}[1]{\mathsf{priority}(#1)}
\newcommand{\rlink}[1]{\mathsf{link}(#1)}
\newcommand{\rsource}[1]{\mathsf{source}(#1)}
\newcommand{\rinterval}[1]{\mathsf{interval}(#1)}
\newcommand{\rlabel}[1]{\mathit{label}[#1]}
\newcommand{\rowner}[2]{\mathit{owner}[#1][#2]}

\newcommand{\MIN}{\mathsf{MIN}}
\newcommand{\MAX}{\mathsf{MAX}}
\newcommand{\CreateAtoms}[1]{\Call{Create\_Atoms}{#1}}
\newcommand{\CreateAtomsPlus}[1]{\Call{Create\_Atoms$^+$}{#1}}

\newcommand{\semantics}[1]{\mbox{$\llbracket \rinterval{#1} \rrbracket$}}

\newtheorem{theorem}{Theorem}

\date{}

\title{\Large \bf Delta-net: Real-time Network Verification Using Atoms}

\author{
{\rm Alex\ Horn}\\
\small Fujitsu Labs of America
\and
{\rm Ali Kheradmand}\\
\small University of Illinois at Urbana-Champaign
\and
{\rm Mukul R. Prasad}\\
\small Fujitsu Labs of America
} 

\begin{document}

\maketitle


\subsection*{Abstract}

Real-time network verification promises to automatically detect violations of network-wide reachability invariants on the data plane. To be useful in practice, these violations need to be detected in the order of milliseconds, without raising false alarms. To date, most real-time data plane checkers address this problem by exploiting at least one of the following two observations: (i)~only small parts of the network tend to be affected by typical changes to the data plane, and (ii)~many different packets tend to share the same forwarding behaviour in the entire network. This paper shows how to effectively exploit a third characteristic of the problem, namely: similarity among forwarding behaviour of packets through \emph{parts} of the network, rather than its entirety. We 
propose the first provably amortized quasi-linear algorithm to do so. We implement our algorithm in a new real-time data plane checker, Delta-net. Our experiments with SDN-IP, a globally deployed ONOS software-defined networking application, and several hundred million IP prefix rules generated using topologies and BGP updates from real-world deployed networks, show that Delta-net checks a rule insertion or removal in approximately 40 microseconds on average, a more than $10\times$ improvement over the state-of-the-art. We also show that Delta-net eliminates an inherent bottleneck in the state-of-the-art that restricts its use in answering Datalog-style ``what if'' queries.

\section{Introduction}
\label{sec:intro}

In an evermore interconnected world, network traffic is increasingly diverse and demanding, ranging from communication between small everyday devices to large-scale data centres across the globe. This diversity has driven the design and rapid adoption of new open networking architectures (e.g.~\cite{SDNSurvey}), built on programmable network switches, which make it possible to separate the control plane from the data plane. This separation opens up interesting avenues for innovation~\cite{McK2011}, including rigorous analysis for finding network-related bugs. Finding these bugs \emph{automatically} poses the following challenges.

Since the control plane is typically a Turing-complete program, the problem of automatically proving the presence and absence of bugs in the control plane is generally undecidable. However, the data plane, which is produced by the control plane, can be automatically analyzed.  While the problem of checking reachability properties in the data plane is generally NP-hard~\cite{MKACGK2011}, the problem becomes polynomial-time solvable in the restricted, but not uncommon, case where network switches only forward packets by matching IP prefixes~\cite{McG2012}. This theoretical fact helps to explain why \emph{real-time data plane checkers}~\cite{KZZCG2013,KCZVMcKW2013,YL2013} can often automatically detect violations of network-wide invariants on the data plane in the order of milliseconds, without raising false alarms.

To achieve this, most real-time network verification techniques exploit at least one of the following two observations: (i)~only small parts of the network tend to be affected by typical changes to the data plane~\cite{KZZCG2013,KCZVMcKW2013}, and (ii)~many different packets often share the same forwarding behaviour in the entire network~\cite{KZZCG2013,YL2013}. Both observations are significant because the former gives rise to \emph{incremental network verification} in which only changes between two data plane snapshots are analyzed, whereas the latter means that the analysis can be performed on a representative subset of network packets in the form of \emph{packet equivalence classes}~\cite{KZZCG2013,KCZVMcKW2013,YL2013}.

In spite of these advances, it is so far an open problem how to efficiently handle operations that involve swaths of packet equivalence classes~\cite{KZZCG2013}. This is problematic because it limits the real-time analysis of network failures, which are common in industry-scale networks, e.g.~\cite{D2008,BK2014}. Moreover, it essentially prevents data plane checkers from being used to answer ``what if'' queries in the style of recent Datalog approaches~\cite{FFPWGMM2015,LBGJV2015} because these hypothetical scenarios typically involve checking the fate of many or all packets in the entire network.


To address this problem, this paper shows how to effectively exploit a third characteristic of data plane checking, namely: similarity among forwarding behaviour of packets through \emph{parts} of the network, rather than its entirety. We show that our approach addresses fundamental limitations (\cref{sec:overview}) in the design of the currently most advanced data plane checker, Veriflow~\cite{KZZCG2013}.

In this paper, we propose a new real-time data plane checker, Delta-net (\cref{sec:delta-net}). Instead of constructing \emph{multiple forwarding graphs} for representing the flow of packets in the network~\cite{KZZCG2013}, Delta-net incrementally transforms a \emph{single edge-labelled graph} that represents \emph{all} flows of packets in the entire network. We present the first provably amortized quasi-linear algorithm to do so (\Cref{theorem:complexity-analysis}). Our algorithm incrementally maintains the lattice-theoretical concept of \emph{atoms}: a set of mutually disjoint ranges through which it is possible to analyze all Boolean combinations of IP prefix forwarding rules in the network so that every possible forwarding table over these rules can be concisely expressed and efficiently checked. This approach is inspired by Yang and Lam's atomic predicates verifier~\cite{YL2013}. While more general, their algorithm has a quadratic worst-case time complexity, whereas ours is quasi-linear. Since Delta-net's atom representation is based on lattice theory, it can be seen as an abstract domain (e.g.~\cite{CC1979}) for analyzing forwarding rules. What makes our abstract domain different from traditional ones is that we dynamically refine its precision so that false alarms never occur. 

For our performance evaluation (\cref{sec:experiments}), we use data sets comprising several hundred million IP prefix rules generated from the UC Berkeley campus, four Rocketfuel topologies~\cite{SMW2002} and real-world BGP updates~\cite{RouteViews}. As part of our experiments, we run SDN-IP~\cite{LHKMKAWB2013,SDNIP}, one of the most mature and globally deployed software-defined networking applications in the ONOS project~\cite{ONOS,ONOSDeployments}. We show that Delta-net checks a rule insertion or removal in tens of microseconds on average, a more than $10\times$ improvement over the state-of-the-art~\cite{KZZCG2013}. Furthermore, as an exemplar of ``what if'' scenarios, we adapt a link failure experiment by Khurshid et al.~\cite{KZZCG2013}, and show that Delta-net performs several orders of magnitude faster than Veriflow~\cite{KZZCG2013}. We discuss related work in~\cref{sec:related-work}.

\vspace{-1em}
\paragraph{Contributions.} Our main contributions are as follows:
\begin{itemize}[leftmargin=*,noitemsep,topsep=0pt,parsep=0pt,partopsep=0pt]
\item Delta-net (\cref{sec:delta-net}), a new real-time data plane checker that incrementally maintains a compact representation about the flows of all packets in the network, thereby supporting a broader class of scenarios and queries.
\item new realistic benchmarks (\cref{subsubsec:ONOS-SDN-IP}) with an open-source, globally deployed SDN application~\cite{SDNIP}.
\item experimental results (\cref{subsec:experimental-results}) that show Delta-net is more than $10\times$ faster than the state-of-the-art in checking rule updates, while also making it now feasible to answer an expensive class of ``what if'' queries.
\end{itemize}

\section{Overview of approach}
\label{sec:overview}

In this section, we motivate and explain our approach through a simple example (\cref{subsec:example}) that illustrates how Delta-net differs from the currently most advanced data plane checker, Veriflow~\cite{KZZCG2013}. In addition to performance considerations, we follow three design goals (\cref{subsec:design-goals}).

\subsection{Example}
\label{subsec:example}

\begin{figure}
\centering
\includegraphics[scale=.8]{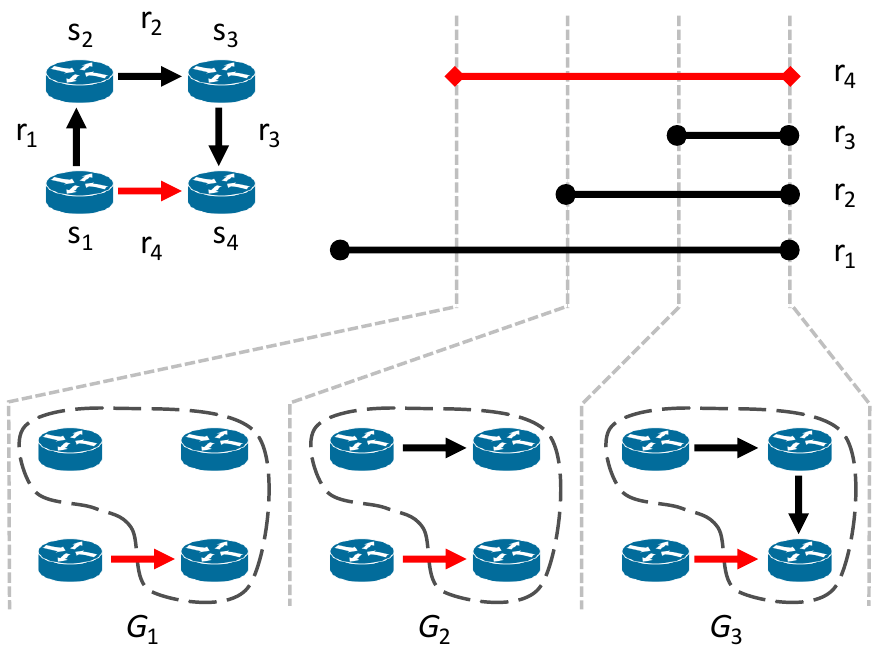}
\caption{When rule $r_4$ (red edge) is inserted into switch $s_1$, Veriflow constructs at least three forwarding graphs, which significantly overlap with each other.}
\label{fig:veriflow}
\end{figure}

Our example is based on a small network of four switches, shown in the upper-left corner of~\Cref{fig:veriflow}. The data plane in this network is depicted as a directed graph in which each edge denotes an IP prefix forwarding rule. For example, rule $r_1$ in~\Cref{fig:veriflow} is assumed to determine the packet flow for a specific destination IP prefix from switch $s_1$ to $s_2$. Suppose the network comprises rules $r_1$, $r_2$ and $r_3$ (black edges) installed on switches $s_1$, $s_2$ and $s_3$, respectively. Since each rule matches packets by a destination IP prefix, we can represent each rule's match condition by an interval. For example, the IP prefix $0.0.0.10/31$ (using the IPv4 CIDR format) corresponds to the half-closed interval $[10:12) = \{10, 11\}$ because $0.0.0.10/31$ is equivalent to the 32-bit binary sequence that starts with all zeros and ends with $101\ast$ where $\ast$ denotes an arbitrary bit. Here, we depict the intervals of all three rules as parallel black lines (in an arbitrary order) in the upper-right half of~\Cref{fig:veriflow}. The interpretation is that all three rules' IP prefixes overlap with each other.

Let us assume we are interested in checking the data plane for forwarding loops. Veriflow then first partitions all packets into \emph{packet equivalences classes}, as explained next. Consider a new rule $r_4$ (red edge in~\Cref{fig:veriflow}) to be installed on switch $s_1$ such that rule $r_4$ has a higher priority than the existing rule $r_1$ on switch $s_1$. As depicted in the upper half of~\Cref{fig:veriflow}, the new rule $r_4$ overlaps with all the existing rules in the network, irrespective of the switch on which they are installed. Veriflow identifies at least three equivalence classes that are affected by the new rule, each of which denotes a set of packets that experience the same forwarding behaviour throughout the network. Here, we depict equivalence classes by three interval segments (gray vertical dashed lines).

For each equivalence class, Veriflow constructs a \emph{forwarding graph} (denoted by $G_1$, $G_2$ and $G_3$ in~\Cref{fig:veriflow}) that represent how packets in each equivalence class can flow through the network. Veriflow can now check for, say, forwarding loops by traversing $G_1$, $G_2$ and $G_3$. Note that the edge that represents the packet flow from switch $s_1$ to $s_2$ is excluded from all three forwarding graphs because on switch $s_1$, for the three depicted equivalence classes, the packet flow is determined by the higher-priority rule $r_4$ rather than the lower-priority rule $r_1$.

Crucially, in our example, the forwarding graphs that Veriflow constructs are essentially the same to previously constructed ones (dashed areas) except for the new edge from switch $s_1$ to $s_4$. In addition, $G_1$, $G_2$ and $G_2$ share much in common, e.g. $G_2$ and $G_3$ have the same edge from switch $s_2$ to $s_3$. As the number of rules in the network increases, so may the commonality among forwarding graphs. In real networks, this leads to inefficiencies that pose problems under real-time constraints.

\begin{figure}
\centering
\includegraphics[scale=.8]{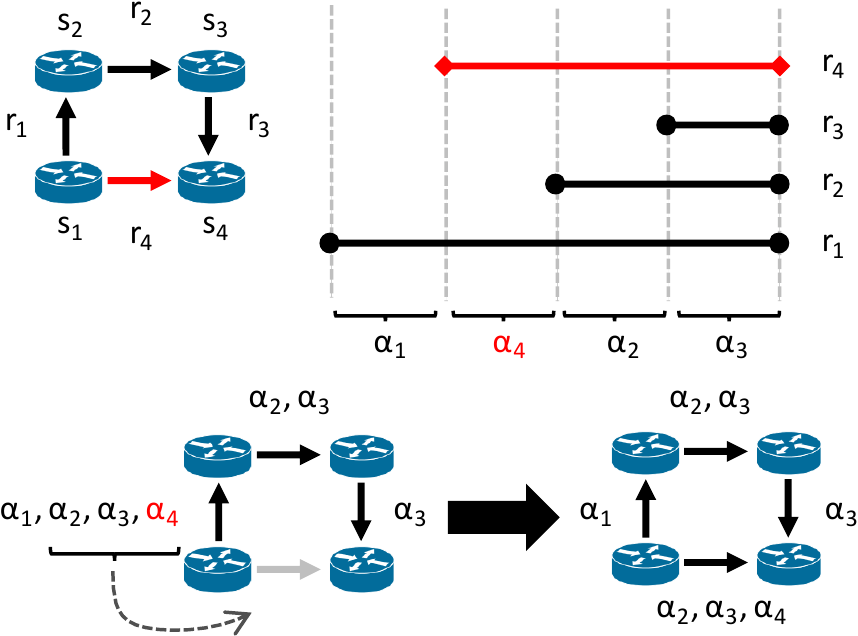}
\caption{Rather than constructing \emph{multiple} forwarding graphs that potentially overlap (\Cref{fig:veriflow}), Delta-net incrementally transforms a \emph{single} edge-labelled graph.}
\label{fig:delta-net}
\end{figure}

We now illustrate how our approach avoids these kind of inefficiencies. For illustrative purposes, assume we start again with the network in which only rules $r_1$, $r_2$ and $r_3$ (black edges) have been installed on switches $s_1$, $s_2$ and $s_3$, respectively. The  collection of IP prefixes in the network induces half-closed intervals, each of which we call an \emph{atom}. A set of atoms can represent an IP prefix. For example, as shown at the top of~\Cref{fig:delta-net}, the set $\{\atom{2}, \atom{3}\}$ represents the IP prefix of rule $r_2$.

At the core of our approach is a directed graph whose edges are labelled by atoms. The purpose of this edge-labelled graph is to represent packet flows in the entire network. For example, to represent that $r_2$ forwards packets from switch $s_2$ to $s_3$ we label the corresponding edge in the directed graph with the atoms $\atom{2}$ and $\atom{3}$.

Of course, an edge-labelled graph that represents all flows in the network may need to be transformed when a new rule is inserted or removed. The bottom of~\Cref{fig:delta-net} illustrates the nature of such a graph transformation in the case where rule $r_4$ is inserted into switch $s_1$. The point of the drawing is threefold. First, observe that the rule insertion of $r_4$ results in the creation of a new atom $\atom{4}$ (red label in the graph on the bottom-left corner). Using the newly created atom, $r_4$'s IP prefix can now be precisely represented as the set of atoms $\{\atom{2}, \atom{3}, \atom{4}\}$. Second, when a new atom, such as $\atom{4}$, is created, existing atom representations may need to be updated. For example, $r_1$'s IP prefix on the edge from switch $s_1$ to $s_2$ needs to be now represented by four instead of only three atoms. Finally, since rule $r_4$, recall, has higher priority than rule $r_1$, three of those four atoms need be moved to the newly inserted edge from switch $s_1$ to $s_4$ (as shown by a dashed arrow in~\Cref{fig:delta-net}). This results in the edge-labelled graph shown in the bottom-right corner of~\Cref{fig:delta-net} where the edges from switch $s_1$ correspond to the forwarding action of the rules $r_1$ and $r_4$ and are labelled by the set of atoms $\{\atom{1}\}$ and $\{\atom{2},\atom{3},\atom{4}\}$, respectively. Crucially, note how our approach avoids the construction of \emph{multiple} overlapping forwarding graphs by transforming a \emph{single} edge-labelled graph instead.

\begin{figure}
\centering
\includegraphics[scale=.317]{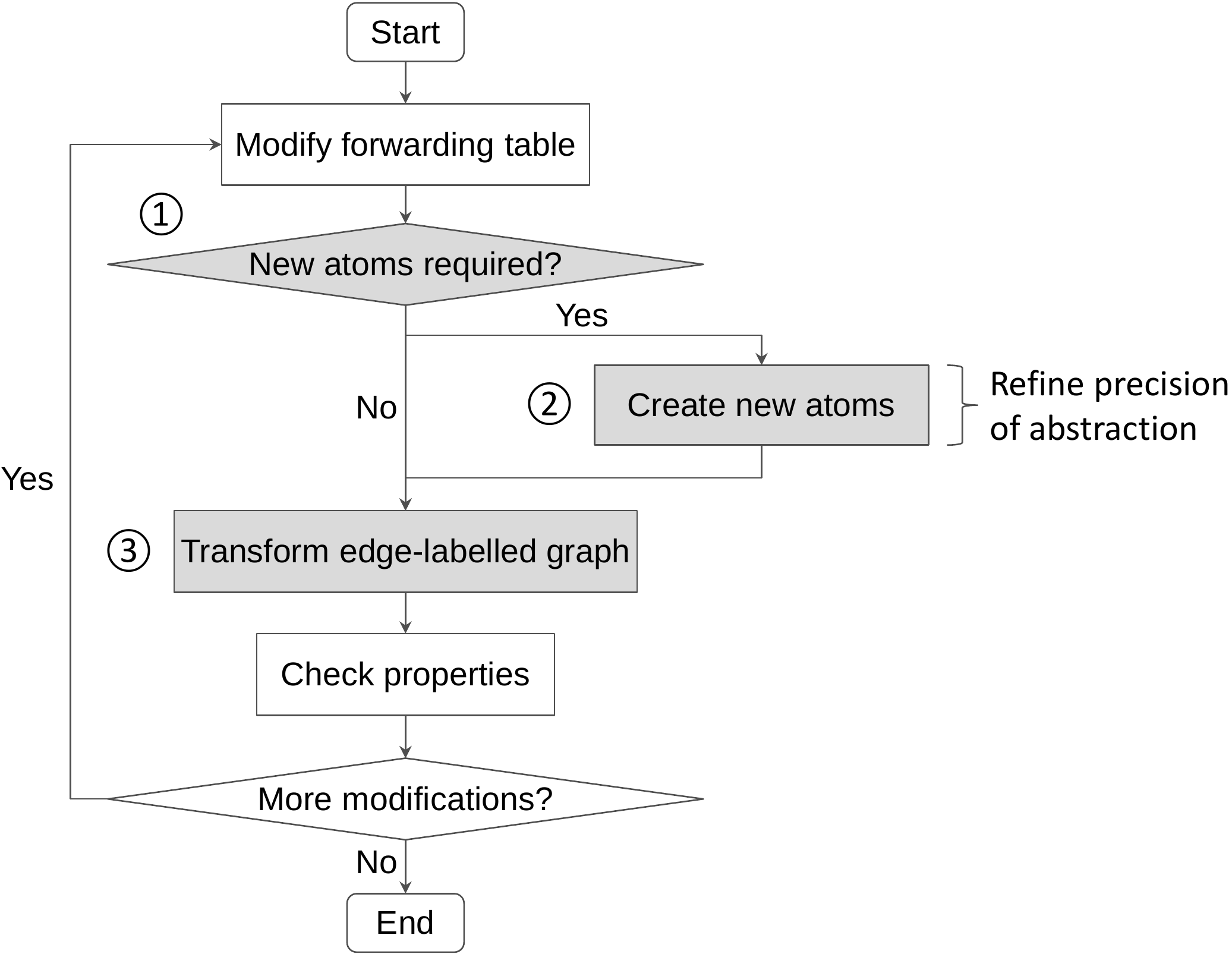}
\caption{Delta-net incrementally maintains atoms, a family of sets of packets that can represent all Boolean combinations of IP prefix forwarding rules.}
\label{fig:flowchart}
\end{figure}

Delta-net's key components and sequence of steps are depicted in~\Cref{fig:flowchart}. In this flowchart, the steps in shaded areas --- annotated by $\{\circled{1},\circled{2}\}$ and $\{\circled{3}\}$ in~\Cref{fig:flowchart} --- are new and described in~\cref{subsec:atom-representation}~and~\cref{subsec:algorithm}, respectively. Here, we only highlight two main fundamental differences between Delta-net and Veriflow:

\begin{itemize}[leftmargin=*,noitemsep,topsep=0pt,parsep=0pt,partopsep=0pt]
\item Veriflow generally has to traverse rules in different switches to compute equivalence classes and forwarding graphs: in our example, when rule $r_4$ is inserted into switch $s_1$, Veriflow traverses all rules in the network (four black edges in~\Cref{fig:veriflow-vs-delta-net-a}). By contrast, our approach concentrates on the affected rules in the modified switch. For example, when rule $r_4$ is inserted into switch $s_1$, the two black edges in~\Cref{fig:veriflow-vs-delta-net-b} show that only rules $r_1$ and $r_4$ on switch $s_1$ are inspected by Delta-net to transform the edge-labelled graph.
\item Veriflow recomputes affected equivalence classes and forwarding graphs each time a rule is inserted or removed, whereas Delta-net incrementally transforms a single edge-labelled graph to represent the flows of \emph{all} packets in the entire network. This significantly broadens the scope of Delta-net (\cref{subsec:design-goals}) because it can more efficiently handle network failures and ``what if'' queries regarding \emph{many or all} packets in the network.
\end{itemize}

\begin{figure}
\centering
\begin{subfigure}{0.48\linewidth}
\centering
\includegraphics[scale=.9]{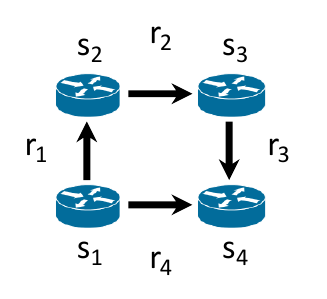}
\caption{Veriflow}
\label{fig:veriflow-vs-delta-net-a}
\end{subfigure}
~
\begin{subfigure}{0.48\linewidth}
\centering
\includegraphics[scale=.9]{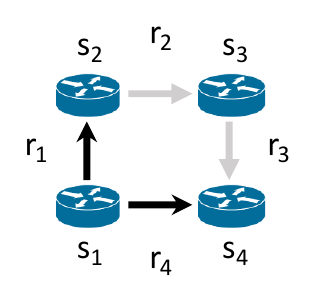}
\caption{Delta-net}
\label{fig:veriflow-vs-delta-net-b}
\end{subfigure}
\caption{Comparison of processed rules (black edges).}
\label{fig:veriflow-vs-delta-net}
\end{figure}



\subsection{Functional design goals}
\label{subsec:design-goals}

\begin{table}[b]
\centering
\begin{tabular}{clr}
  \toprule
    \textbf{Priority} & \multicolumn{1}{c}{\textbf{IP Prefix}} & \multicolumn{1}{c}{\textbf{Action}} \\
    \midrule
    High      & $0.0.0.10/31$ & drop    \\
    Low       & $0.0.0.0/28$  & forward \\
  \bottomrule
\end{tabular}
\caption{A forwarding table for a network switch.}
\label{table:forwarding-rules}
\end{table}

In addition to more stringent real-time constraints, our work is guided by the following three design goals:
\begin{enumerate}[leftmargin=*,noitemsep,topsep=3pt,parsep=0pt,partopsep=0pt]
\item Similar to Datalog-based approaches~\cite{FFPWGMM2015,LBGJV2015}, we want to efficiently find \emph{all} packets that can reach a node $B$ from $A$, avoiding restrictions of SAT/SMT-based data plane checkers (e.g.~\cite{MKACGK2011}), which can solve a broader class of problems but require multiple calls to their underlying SAT/SMT solver to find more than one witness for the reachability from $A$ to $B$.
\item Our design should support known incremental network verification techniques that construct forwarding graphs for the purpose of checking reachability properties each time a rule is inserted or removed~\cite{KZZCG2013}. This is important because it preserves one of the main characteristics of previous work, namely: it is practical, and no expertise in formal verification is required to check the data plane.
\item When real-time constraints are less important (as in the case of pre-deployment testing, e.g.~\cite{ZKVM2012}), we want to facilitate the answering of a broader class of (possibly incremental) reachability queries, such as \emph{all-pairs reachability} queries in the style of recent Datalog approaches~\cite{FFPWGMM2015,LBGJV2015}. These kind of queries generally concern the reachability between \emph{all} packets and pairs of nodes in the network. We also aim at efficiently answering queries in scenarios that involve many or all packets, such as link failures~\cite{KZZCG2013}.
\end{enumerate}

After explaining the technical details of Delta-net, we describe how it achieves these design goals (\cref{subsec:revisited-design-goals}).

\section{Delta-net}
\label{sec:delta-net}

In this section, we explain Delta-net's underlying atom representation (\cref{subsec:atom-representation}), and its algorithm for modifying rules through insertion and removal operations (\cref{subsec:algorithm}). Recall that these two subsections correspond to the steps annotated by $\{\circled{1},\circled{2}\}$ and $\{\circled{3}\}$ in~\Cref{fig:flowchart}, respectively.

We illustrate the internal workings of Delta-net using the simple forwarding table in~\Cref{table:forwarding-rules}. It features two rules, $r_H$ and $r_L$, whose subscript corresponds to their priority: the higher-priority rule, $r_H$, drops packets whose destination address matches the IP prefix $0.0.0.10/31$, whereas the lower-priority rule, $r_L$, forwards packets destined to the IP prefix $0.0.0.0/28$. We elide details about the next hop (where a matched packet should be sent) because it is not pertinent to the example.

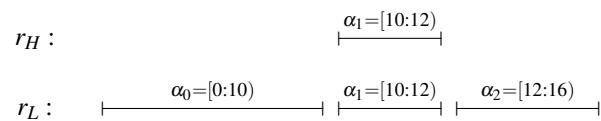
\begin{figure}[b]
\xymatrix@R=1.45em@C=0.89em{
 r_H: &&&&&&& \ar@{|-|}[rrr]^{\atom{1} = [10:12)} &&&&&&& \\
 r_L: &\ar@{|-|}[rrrrrr]^{\atom{0} = [0:10)} &&&&&&\ar@{|-|}[rrr]^{\atom{1} = [10:12)}&&&\ar@{|-|}[rrrr]^{\atom{2} = [12:16)}&&&&
}
\caption{Atoms for the IP prefix rules in~\Cref{table:forwarding-rules}.}
\label{fig:half-closed-intervals}
\end{figure}

As alluded to in the previous section (\cref{subsec:example}), we can think of IP prefixes as half-closed intervals: $r_H$'s IP prefix, $0.0.0.10/31$, corresponds to the half-closed $[10:12)$. Similarly, $0.0.0.0/28\, =\, [0:16)$ for $r_L$'s IP prefix. Of course, this interval representation can be easily generalized to IPv6 addresses. Next, we show how Delta-net represents rules with such IP prefixes, for some fixed IP address length.

\subsection{Atom representation}
\label{subsec:atom-representation}

In this subsection, we describe the concept of atoms; how they are maintained is essential to the rule modifications algorithms in the next subsection (\cref{subsec:algorithm}).

Intuitively, we can segment the IP prefixes of all the rules in the network into disjoint half-closed intervals, which we call atoms. This kind of segmentation is illustrated in~\Cref{fig:half-closed-intervals} using the rules $r_H$ and $r_L$ in~\Cref{table:forwarding-rules}.\footnote{Appendix~\ref{appendix:Hasse-diagram} illustrates the fact that atoms induce a Boolean lattice.}

By construction of atoms, we can represent an IP prefix of a rule $r$ as a \emph{set of atoms}. We denote this IP prefix representation by $\semantics{r}$. For example, $r_H$'s IP prefix, $\semantics{r_H}$, corresponds to the singleton set consisting of the atom $\atom{1}$, whereas $r_L$'s IP prefix is $\semantics{r_L} = \{\atom{0},\atom{1},\atom{2}\}$. Using these atoms, we can represent, for example, the set difference $\semantics{r_L} - \semantics{r_H}$ to formalize the fact that $r_L$ can only match packets that are not matched by the higher-priority rule $r_H$. Next, we explain how to devise an efficient representation of atoms such that we can efficiently verify network-wide reachability properties when a rule is inserted or removed (\cref{subsec:algorithm}).

At the core of our atom representation is a function, $\bounds$, that maps non-negative integers to identifiers. Specifically, $\bounds$ is an ordered map that contains key/value pairs $n \mapsto \atom{i}$ where $n$ is a lower or upper bound of an IP prefix of a rule $r$ (denoted by $\rlower{r}$ and $\rupper{r}$, respectively) and $\atom{i}$ is a unique identifier, called \emph{atom identifier}. For example, $\rlower{r_H} = 10$ and $\rupper{r_H} = 12$. More generally, we ensure that $\MIN \leq \rlower{r} < \rupper{r} \leq \MAX$ for every rule $r$ where $\MIN = 0$ and $\MAX = 2^k$ for some fixed positive integer $k$, e.g. $k = 32$ for 32-bit IP addresses. We maintain the invariant that $\bounds$ contains only unique keys. The interpretation of each pair $n \mapsto \atom{i}$ in $\bounds$, for all $n < \MAX$, is as follows: the atom identifier $\atom{i}$ denotes the \emph{atom} $[n:n')$ where $n'$ is the next numerically greater key in $\bounds$. Each atom identifier, therefore, uniquely denotes a half-closed interval, i.e. an atom. For efficiency reasons, we ensure that each atom identifier is generated from a consecutively increasing counter that starts at zero. Before processing any rules, we initialize $\bounds$ by inserting $\MIN \mapsto \atom{0}$ and $\MAX \mapsto \atom{\infty}$ where $\atom{\infty}$ is the greatest atom identifier.

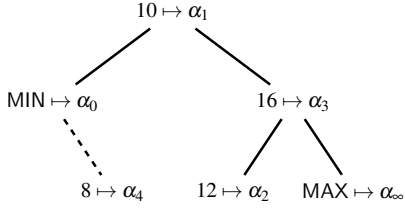
\begin{figure}
\centering
\scalebox{0.8}{
\begin{tikzpicture}[very thick,level/.style={sibling distance=40mm/#1}]
\node [vertex] (r){$10 \mapsto \atom{1}$}
  child {
    node [vertex] (a) {$\MIN \mapsto \atom{0}$}
    child[missing]
    child {
      node [vertex] {$8 \mapsto \atom{4}$} edge from parent[-, dashed]
    }
  }
  child {
    node [vertex] {$16 \mapsto \atom{3}$}
    child {
      node [vertex] {$12 \mapsto \atom{2}$}
    }
    child {
      node [vertex] {$\MAX \mapsto \atom{\infty}$}
    }
  };
\end{tikzpicture}
}
\caption{Balanced binary search tree of key/value pairs after inserting the half-closed intervals from~\Cref{fig:half-closed-intervals}.}
\label{fig:bst}
\end{figure}

We define the procedure \CreateAtoms{$r$}, where $\rinterval{r} = [\rlower{r}:\rupper{r})$ is the half-closed interval corresponding to $r$'s IP prefix, such that, if $\bounds$ has not already paired $\rlower{r}$ with an atom identifier, then it inserts into $\bounds$ the key/value pair $\rlower{r} \mapsto \atom{j}$ for the next available counter value $\atom{j}$; similarly, we conditionally insert into $\bounds$ the key/value pair $\rupper{r} \mapsto \atom{k}$ for the next available counter value $\atom{k}$. Note that after \CreateAtoms{$r$} has been called, $\bounds$ may contain $0$, $1$, or $2$ new atoms (but not more). For example, IP prefixes such as $1.2.0.0/16$ and $1.2.0.0/24$ have the same lower bound because they only differ in their prefix lengths, and so together yield only three and not four atoms. While the values of atom identifiers depend on the order in which rules are inserted, the set of generated atoms at the end is invariant under the order in which \textproc{Create\_Atoms} is called. We also remark that the number of atoms represented by $\bounds$ is equal to $\bounds$'s size minus one. 

For our complexity analysis, we assume that the $\bounds$'s insertion and retrieval operations run logarithmically in the size of $\bounds$, which could be achieved with a balanced binary-search tree such as a red-black tree. In this case, \Cref{fig:bst} (excluding the leaf node connected by a dashed edge) illustrates the balanced binary search tree that results after \CreateAtoms{$r_H$} and \CreateAtoms{$r_L$} has been called for the rules $r_H$ and $r_L$ in~\Cref{table:forwarding-rules}. For example, $\atom{1}$ at the root of the binary search tree in~\Cref{fig:bst} denotes the atom $[10:12)$. When clear from the context, we refer to atom identifiers and atoms interchangeably.

\subsection{Edge labelling algorithm}
\label{subsec:algorithm}

Using our atom representation (\cref{subsec:atom-representation}), we show how to efficiently label the edges of a directed graph that succinctly describes the flow of all packets in the entire network. Our algorithm is incremental in the sense that it only changes edge labels that are affected by the insertion or removal of a rule. Our algorithm, which achieves this incrementality, requires the following notions.

We denote an IP prefix forwarding rule by $r$, possibly with a prime symbol. Each rule $r$ is associated with $\rpriority{r}$ and $\rlink{r}$, as explained in turn. We assume that rules in the same forwarding table whose IP prefixes overlap have pair-wise distinct priorities, denoted by $\rpriority{r}$.\footnote{This assumption is reasonable for, say, OpenFlow tables where the matching of rules with the same highest priority is explicitly undefined.} For all rules $r$ and $r'$ in the same forwarding table, $r$ has a \emph{higher priority than} $r'$ if $\rpriority{r} > \rpriority{r'}$; equivalently, $\rpriority{r} < \rpriority{r'}$ means that $r$ has a \emph{lower priority than} $r'$. Note that longest-prefix routing can be simulated by assigning rule priorities according to prefix lengths~\cite{YL2013}. We denote by $\rlink{r}$ a directed edge in a graph that is induced by a network topology. For theoretical and practical reasons (see also~\cref{subsec:implementation}), $\rlink{r}$ is purposefully more general than a pair of, say, ports. We write $\rsource{r}$ for the node in the graph on which $\rlink{r}$ is incident. For example, $\rsource{r_1} = s_1$ and $\rsource{r_2} = s_2$ in~\Cref{fig:delta-net}.

From a high-level perspective, Delta-net consists of two algorithms, one for inserting (\Cref{alg:insert-rule}) and another for removing (\Cref{alg:remove-rule}) a single rule. Both algorithms accesses three global variables: $\bounds$, $\mathit{label}$ and $\mathit{owner}$, as described in turn. First, $\bounds$ is the balanced binary tree described in~\cref{subsec:atom-representation}, e.g.~\Cref{fig:bst}. Second, given a $\mathit{link}$ in the network topology, $\rlabel{\mathit{link}}$ denotes a set of atoms, each of which corresponds to a half-closed interval that a designated field in a packet header $h$ can match for $h$ to be forwarded along the $\mathit{link}$. Finally, $\mathit{owner}$ is an array of hash tables, each of which stores a balanced binary search tree containing rules ordered by priority. More accurately, $\mathit{owner}$ is an array of sufficient size such that, for every atom $\alpha$, $\mathit{owner}[\alpha]$ is a hash table that maps a $\mathit{source}$ node to a balanced binary search tree, $\mathit{bst}$, that orders rules in the $\mathit{source}$ node that contain atom $\alpha$ in their interval according to their priority, i.e., we maintain the invariant that $\mathit{bst}$ contains only rules $r$ such that $\mathit{source} = \rsource{r}$ and $\alpha \in \semantics{r}$ where $\mathit{bst} = \rowner{\alpha}{\mathit{source}}$. The highest-priority rule in a non-empty balanced binary search tree $\mathit{bst}$ can be retrieved via $\mathit{bst}.\mathsf{highest\_priority\_rule}()$. We remark that we do not use a priority queue because~\Cref{alg:remove-rule} described later (\cref{subsubsec:remove}) needs to be able to remove arbitrary rules, not just the highest-priority one. We write $r \in \mathit{bst}$ when rule $r$ is stored in $\mathit{bst}$.

\begin{algorithm}[t]
\begin{algorithmic}[1]
\Procedure{Insert\_Rule}{$r$}
\State $\Delta \gets \CreateAtomsPlus{r}$ \Comment{$\abs{\Delta} \leq 2$}\label{line:insert-rule-split-atoms-begin}
\For{$\alpha \mapsto \alpha'$ \textbf{in} $\Delta$}\label{line:insert-rule-split-atoms-loop}
  \State $\mathit{owner}[\alpha'] \gets \mathit{owner}[\alpha]$\label{line:owner-source-r-prime}
  \For{$\mathit{source} \mapsto \mathit{bst}$ \textbf{in} $\mathit{owner}[\alpha]$}
    \State $r' \gets \mathit{bst}.\mathsf{highest\_priority\_rule}()$\label{line:insert-rule-get-r-prime}
    \State $\rlabel{\rlink{r'}} \gets \rlabel{\rlink{r'}} \cup \{\alpha'\}$\label{line:update-atoms}
  \EndFor
\EndFor\label{line:insert-rule-split-atoms-end}
\For{$\alpha$ \textbf{in} $\semantics{r}$}\label{line:insert-rule-atom-loop-begin}
  \State $r' \gets \mathbf{null}$
  \State $\mathit{bst} \gets \rowner{\alpha}{\rsource{r}}$\label{line:insert-rule-get-bst}
  \If{$\mathbf{not}\ \mathit{bst}.\mathsf{is\_empty}()$}
    \State $r' \gets \mathit{bst}.\mathsf{highest\_priority\_rule}()$\label{line:insert-rule-get-highest-priority-rule}
  \EndIf
  \If{$r' = \mathbf{null}\ \mathbf{or}\ \rpriority{r'} < \rpriority{r}$} \label{line:highest-priority-rule-check-begin}
    \State $\rlabel{\rlink{r}} \gets \rlabel{\rlink{r}} \cup \{\alpha\}$\label{line:update-r-guard}
    \If{$r' \not= \mathbf{null}\ \mathbf{and}\ \rlink{r} \not= \rlink{r'}$}
      \State $\rlabel{\rlink{r'}} \gets \rlabel{\rlink{r'}} - \{\alpha\}$
    \EndIf\label{line:update-r-prime-guard}
  \EndIf\label{line:highest-priority-rule-check-end}
  \State $\mathit{bst}.\mathsf{insert}(r)$\label{line:owner-source-r}
\EndFor\label{line:insert-rule-atom-loop-end}
\EndProcedure \label{line:insert-rule-end}
\end{algorithmic}
\caption{Inserts rule $r$ into a forwarding table.}
\label{alg:insert-rule}
\end{algorithm}

\subsubsection{Edge labelling when inserting a rule}
\label{subsubsec:insert}

We now explain how the \textproc{Insert\_Rule} procedure in~\Cref{alg:insert-rule} works. The algorithm starts by calling \textproc{Create\_Atoms$^+$} (\cref{line:insert-rule-split-atoms-begin}) that accomplishes the same as \textproc{Create\_Atoms} from~\cref{subsec:atom-representation} except that \textproc{Create\_Atoms$^+$} also returns $\Delta$, a set of \emph{delta-pairs}, as explained next. Each delta-pair in $\Delta$ is of the form $\alpha \mapsto \alpha'$ where $\alpha$ and $\alpha'$ are atoms. The intuition is that the half-closed interval previously represented by $\alpha$ needs to be now represented by two atoms instead, namely $\alpha$ and $\alpha'$. We call this \emph{atom splitting}. In a nutshell, this splitting provides an efficient mechanism for incrementally refining the precision of our abstract domain. This incremental abstraction refinement allows us to precisely and efficiently represent all Boolean combinations of rules in the network (see also~\cref{sec:intro}).

To illustrate the splitting of atoms, let $r_M$ be a new medium-priority rule to be inserted into~\Cref{table:forwarding-rules} such that $\rpriority{r_L} < \rpriority{r_M} < \rpriority{r_H}$. Assume $r_M$'s IP prefix is $0.0.0.8/30$; hence, $\rinterval{r_M} = [8 : 12)$. If $\bounds$ is the binary search subtree in~\Cref{fig:bst} consisting of undashed edges, then $\CreateAtomsPlus{r_M}$ returns a single delta-pair, namely $\Delta = \{\atom{0} \mapsto \atom{4}\}$, where $\atom{0}$ is the atom identifier denoting the atom $[\MIN : 10)$ before $r_M$ has been inserted, and $\atom{4}$ is a new atom identifier, depicted as a dashed leaf in~\Cref{fig:bst}. Here, $\Delta = \{\atom{0} \mapsto \atom{4}\}$ means that the existing atom $[\MIN : 10)$ needs to be split into $\atom{0} = [\MIN : 8)$ and $\atom{4} = [8 : 10)$. Note that there are always at most two delta-pairs in $\Delta$. Thus, since $\abs{\Delta} \leq 2$, we can effectively update the atom representation of forwarding rules in an incremental manner.

The splitting of atoms is effectuated by updating the labels for some links in the single-edged graph that represents the flow in the entire network (\cref{line:update-atoms}). To quickly determine these links, we exploit the highest-priority matching mechanism of packets. For this purpose, we use the array of hash tables, $\mathit{owner}$: it associates an atom $\alpha$ and $\mathit{source}$ node with a binary search tree $\mathit{bst}$ such that $\mathit{bst}.\mathsf{highest\_priority\_rule}()$ determines the next hop from $\mathit{source}$ of an $\alpha$-packet (\cref{line:insert-rule-get-r-prime}). Since $\abs{\Delta} \leq 2$, the doubly nested loop (\cref{line:insert-rule-split-atoms-loop}--\ref{line:insert-rule-split-atoms-end}) runs at most twice. For each delta-pair $\alpha \mapsto \alpha'$ in $\Delta$, the array of hash tables is updated so that $\mathit{owner}[\alpha']$ is a copy of $\mathit{owner}[\alpha]$ (\cref{line:owner-source-r-prime}). Therefore, since $r' \in \rowner{\alpha}{\rsource{r'}}$ holds for the existing atom $\alpha$, it follows that $r' \in \rowner{\alpha'}{\rsource{r'}}$ holds for the new atom $\alpha' \in \semantics{r'}$, thereby maintaining the invariant of the $\mathit{owner}$ array of hash tables (\cref{subsec:algorithm}). We adjust the labels accordingly (\cref{line:update-atoms}). The remainder of \Cref{alg:insert-rule} (\cref{line:insert-rule-atom-loop-begin}--\ref{line:insert-rule-atom-loop-end}) reassigns atoms based on the priority of the rule that `owns' each atom, as explained next.

The algorithm continues by iterating over all atoms that collectively represent $r$'s IP prefix (\cref{line:insert-rule-atom-loop-begin}), possibly including the newly created atom(s) in $\Delta$ (see previous paragraphs). For each such atom $\alpha$ in $\semantics{r}$, we find the highest-priority rule $r'$ (\cref{line:insert-rule-get-highest-priority-rule}) that determines the flow of an $\alpha$-packet at the node $\rsource{r}$ into which rule $r$ is inserted. We say such a rule $r'$ \emph{owns} $\alpha$. If no such rule exists or its priority is lower than $r$'s (\cref{line:highest-priority-rule-check-begin}), we assign $\alpha$ to the set of atoms that determine which network traffic can flow along the link of $r$ (\cref{line:update-r-guard}--\ref{line:update-r-prime-guard}), i.e. $\rlabel{\rlink{r}}$. Finally, we insert $r$ into the binary search tree for atom $\alpha$ and node $\rsource{r}$ (\cref{line:owner-source-r}), irrespective of which rule owns atom $\alpha$.

\subsubsection{Edge labelling when removing a rule}
\label{subsubsec:remove}

\Cref{alg:remove-rule} removes a rule $r$ from a forwarding table. Similar to~\Cref{alg:insert-rule}, \Cref{alg:remove-rule} iterates over all atoms $\alpha$ that are needed to represent $r$'s IP prefix (\cref{line:remove-rule-begin}). For each such atom $\alpha$, it retrieves the $\mathit{bst}$ that is specific to the node from which $r$ should be removed (\cref{line:remove-rule-bst}). After finding the highest-priority rule $r'$ in $\mathit{bst}$ (\cref{line:remove-rule-prime}), it removes $r$ from $\mathit{bst}$ (\cref{line:remove-rule-r}). If $r'$ equals $r$ (\cref{line:remove-rule-check-prime}), we need to remove $\alpha$ from the label of $\rlink{r}$ because the rule that needs to be removed, $r$, owns atom $\alpha$ (as described in~\cref{subsubsec:insert}). In addition, we may need to transfer the ownership of the next higher priority rule (\cref{line:remove-rule-check-empty-begin}-\ref{line:remove-rule-check-empty-end}). 

We remark that after the removal of a rule, it may be that some (at most two) atoms are not needed any longer. In this case, akin to garbage collection, we could reclaim the unused atom identifier(s). This `garbage collection' mechanism is omitted from \Cref{alg:remove-rule}.

\begin{algorithm}[t]
\begin{algorithmic}[1]
\Procedure{Remove\_Rule}{$r$}
\For{$\alpha$ \textbf{in} $\semantics{r}$}\label{line:remove-rule-begin}
  \State $\mathit{bst} \gets \rowner{\alpha}{\rsource{r}}$\label{line:remove-rule-bst}
  \State $r' \gets \mathit{bst}.\mathsf{highest\_priority\_rule}()$\label{line:remove-rule-prime}
  \State $\mathit{bst}.\mathsf{remove}(r)$\label{line:remove-rule-r}
  \If{$r' = r$}\label{line:remove-rule-check-prime}
    \State $\rlabel{\rlink{r}} \gets \rlabel{\rlink{r}} - \{\alpha\}$
    \If{$\mathbf{not}\ \mathit{bst}.\mathsf{is\_empty}()$}\label{line:remove-rule-check-empty-begin}
      \State $r'' \gets \mathit{bst}.\mathsf{highest\_priority\_rule}()$
      \State $\rlabel{\rlink{r''}} \gets \rlabel{\rlink{r''}} \cup \{\alpha\}$
    \EndIf\label{line:remove-rule-check-empty-end}
  \EndIf
\EndFor\label{line:remove-rule-end}
\EndProcedure
\end{algorithmic}
\caption{Removes rule $r$ from a forwarding table.}
\label{alg:remove-rule}
\end{algorithm}

\subsubsection{Complexity analysis}

We now show that each rule update is amortized linear time in the number of affected atoms and logarithmic in the maximum number of overlapping rules in a single switch. While in the worst-case there are as many atoms as there are rules in the network, our experiments (\cref{sec:experiments}) show that the number of atoms is typically much smaller in practice, explaining why we found Delta-net to be highly efficient in the vast majority of cases.

\begin{theorem}[Asymptotic worst-case time complexity]
To insert or remove a total of $R$ rules, \Cref{alg:insert-rule}~and~\ref{alg:remove-rule} have a $O(RK \log M)$ worst-case time complexity where $K$ is the number of atoms and $M$ is the maximum number of overlapping rules per network switch.
\label{theorem:complexity-analysis}
\end{theorem}
\begin{proof}
The proof can be found in Appendix~\ref{appendix:proof-of-complexity-analysis}.
\end{proof}

The space complexity of Delta-net is $O(RK)$ where $R$ and $K$ are the total number of rules and atoms, respectively. We recall that $K$ is significantly smaller than $R$. We also experimentally quantify memory usage (\cref{sec:experiments}).

\subsection{Revisited: functional design goals}
\label{subsec:revisited-design-goals}

From a functionality perspective, recall that our work is guided by three design goals (\cref{subsec:design-goals}). In this subsection, we explain how Delta-net achieves these goals.

\paragraph{API for persistent network-wide flow information.} Delta-net provides an exact representation of all flows through the entire network. For this purpose, Delta-net maintains the atom labels for every edge in the graph that represents the network topology. From a programmer's perspective, this edge-centric information can be always retrieved in constant-time through $\rlabel{\mathit{link}}$ where $\mathit{link}$ is a pair of nodes in this graph. This way, our API allows a programmer to answer reachability questions about packet flow through the entire network irrespective of the rule that has been most recently inserted or removed. This makes Delta-net different from Veriflow~\cite{KZZCG2013}. Architecturally, our generalization is achieved by decoupling packet equivalence classes (whether affected by a rule update or not) from the construction of their corresponding forwarding graphs, cf.~\cite{KZZCG2013}.

\vspace{-1em}
\paragraph{Incremental network verification via delta-graphs.} Similar to Veriflow~\cite{KZZCG2013}, Delta-net can build forwarding graphs, if necessary, to check reachability properties that are suitable for incremental network verification, such as checking the existence of forwarding loops each time a rule is inserted or removed. In fact, the concept of atoms has as consequence a convenient algorithm for computing a compact edge-labelled graph, called \emph{delta-graph}, that represents all such forwarding graphs. We can generate a delta-graph as a by-product of~\Cref{alg:insert-rule} for all atoms $\alpha$ whose owner changes (\cref{line:highest-priority-rule-check-begin}-\ref{line:highest-priority-rule-check-end}); similarly for~\Cref{alg:remove-rule}. If so desired, multiple rule updates may be aggregated into a delta-graph.

\begin{algorithm}[t]
\begin{algorithmic}[1]
\For{$k, i, j$ \textbf{in} $V$} \Comment{Triple nested loop}
  \State $\rlabel{i, j} \gets \rlabel{i, j}\,\cup\,(\rlabel{i, k}\,\cap\,\rlabel{k, j})$
\EndFor
\end{algorithmic}
\caption{Compute all-pairs reachability of all atoms.}
\label{alg:all-pairs-reachability}
\end{algorithm}

\paragraph{Easier checking of other reachability properties.} Delta-net's design provides a lattice-theoretic foundation for transferring known algorithmic techniques to the field of network verification. For example,~\Cref{alg:all-pairs-reachability} adapts the Floyd–Warshall algorithm to compute the transitive closure of packet flows between all pairs of nodes in the network. Note that our adaptation interchanges the usual maximum and addition operators with union and intersection of sets of atoms, respectively. This way,~\Cref{alg:all-pairs-reachability} process multiple packet equivalence classes in each hop.\footnote{A routine proof by induction on $k$ (the outermost loop) shows that~\Cref{alg:all-pairs-reachability} computes the all-pairs reachability of every $\alpha$-packet.} Veriflow has not been designed for such computations, and~\Cref{alg:all-pairs-reachability} illustrates how Delta-net facilitates use cases beyond the usual reachability checks, cf.~\cite{KZZCG2013,KCZVMcKW2013,YL2013}. This algorithm could be run either on the edge-labelled graph that represents the entire network or only its incremental version in form of a delta-graph (see previous paragraph).

While decision problems such as all-pairs reachability have a higher computational complexity (e.g.,~\Cref{alg:all-pairs-reachability}'s complexity is $O(K\abs{V}^3)$ where $K$ and $V$ is the number of atoms and nodes in the edge-labelled graph, respectively), they are relevant and useful during pre-deployment testing of SDN applications, as demonstrated by recent work on Datalog-based network verification, e.g.~\cite{FFPWGMM2015,LBGJV2015}. The fact that our design makes it possible to verify network-wide reachability by intersecting or taking the union of sets of atoms~\cite{YL2013} is also relevant for scenarios that involve many or all packet equivalence classes at a time, such as ``what if'' queries, network failures, and traffic isolation properties, e.g.~\cite{AFGJZSW2014,FKMST2015}.


\section{Performance evaluation}
\label{sec:experiments}

In this section, we experimentally evaluate our implementation of Delta-net (\cref{subsec:implementation}) on a diverse range of data sets (\cref{subsec:data-sets}) that are significantly larger than previous ones (see also Appendix~\ref{appendix:data-sets}). Our experiments provide strong evidence that Delta-net significantly advances the field of real-time network verification (\cref{subsec:experimental-results}).

\subsection{Implementation}
\label{subsec:implementation}

We implemented~\Cref{alg:insert-rule}~and~\ref{alg:remove-rule} in C++14~\cite{CPP14}. Our implementation is single-threaded and comprises around 4,000 lines of code that only depend on the C++14 standard library. In particular, we use the standard hashmap, balanced binary search tree and resizeable array implementations. We implement edge labels as customized dynamic bitsets, stored as aligned, dynamically allocated, contiguous memory. We detect forwarding loops via an iterative depth-first graph traversal.

We remark that while \Cref{alg:insert-rule}~and~\ref{alg:remove-rule} focus on handling IP prefix rules, our approach can be extended for other packet header fields. For non-wildcard (i.e. concrete) header fields, our implementation achieves this by encoding composite match conditions as separate nodes in the single edge-labelled graph. For example, if a switch $s$ contains rules that can match three input ports, we encode $s$ as three separate nodes in the edge-labelled graph. It is for this reason that we report the number of graph nodes rather than the number of switches when describing our data sets in the next subsection.

\subsection{Description of data sets}
\label{subsec:data-sets}

Our data sets are publicly available~\cite{Deltanet} and can be broadly divided into two classes: data sets derived from the literature (\cref{subsubsec:synthetic}), and data sets gathered from an ONOS SDN  application (\cref{subsubsec:ONOS-SDN-IP}). Both are significant as the former avoids experimental bias, whereas the latter increases the realism of our experiments. To achieve reproducibility, we organize our data sets as text files in which each line denotes an \emph{operation}: an insertion or removal of a rule. So all operations can be easily replayed.

\Cref{table:data-sets} summarizes our data sets in terms of three metrics.
The second and third column in~\Cref{table:data-sets} correspond to the maximum number of nodes and links in the edge-labelled graph, respectively. We recall that the number of nodes is proportional to the number of ports and switches in the network (\cref{subsec:implementation}). The total number of operations is reported in the last column. Note that most of our data sets are significantly larger than previous ones, cf.~\cite{KZZCG2013,C2014,KCZVMcKW2013,YL2013} (see also Appendix~\ref{appendix:data-sets}). Next, we describe the main features of our data sets.

\begin{table}
\centering
\scalebox{0.76}{
\begin{tabular}{l|r|r|r}
  \toprule
    \multicolumn{1}{c|}{\small\textbf{Data set}} & \small\textbf{Nodes} & \small\textbf{Max Links} & \small\textbf{Operations} \\
    \midrule
    Berkeley          & $23$  & $252$     & $25.6 \times 10^6$        \\
    INET              & $316$ & $40,770$ & $249.5 \times 10^6$        \\
    RF~1755           & $87$  & $2,308$  & $67.5 \times 10^6$       \\
    RF~3257           & $161$ & $9,432$ & $149.0 \times 10^6$        \\
    RF~6461           & $138$ & $8,140$  & $150.0 \times 10^6$      \\
    Airtel~1          & $68$  & $260$     & $14.2 \times 10^6$ \\
    Airtel~2          & $68$  & $260$    & $505.2 \times 10^6$ \\
    4Switch           & $12$  & $16$     & $1.12 \times 10^6$ \\
  \bottomrule
\end{tabular}
}
\caption{Data sets used for evaluating Delta-net.}
\label{table:data-sets}
\end{table}

\subsubsection{Synthetic data sets}
\label{subsubsec:synthetic}

To avoid experimental bias, our experiments purposefully include data sets from the literature~\cite{ZZYJJLMV2014,NTRW2016} that feature network topologies from the UC Berkeley campus and the Rocketfuel (RF) project~\cite{SMW2002}, namely ASes 1755,~1239,~6257~and~6461. Note that the RF topologies in~\cite{NTRW2016} correspond to those used by~\cite{HVSBFTF2015,Frenetic,T2016}. For each of these five network topologies, we generate forwarding rules following the same mechanism as in~\cite{ZZYJJLMV2014}, namely: we gather IP prefixes from over a half a million of real-world BGP updates collected by the Route Views project~\cite{RouteViews}, and compute the shortest paths in a network topology~\cite{Libra}. For example, for the network topology RF~1239, this results in the INET data set~\cite{ZZYJJLMV2014}, a synthetic wide-area backbone network that contains approximately $300$ routers, $481$ thousand subnets and $125$ million IPv4 forwarding rules. We modify the data sets so that rules are inserted with a random priority. After rules have been inserted, we remove them in random order. The first five rows in \Cref{table:data-sets} show the resulting data sets, which contain up to 125 million rules. Due to rule removals, the total number of operations is twice the maximum number of rules. Collectively, the Berkeley, INET and RF~1755,~3257~and~6461 data sets comprise around 640 million rule operations. Next, we explain the remaining three data sets in~\Cref{table:data-sets}.

\subsubsection{SDN-IP Application}
\label{subsubsec:ONOS-SDN-IP}

In addition to synthetic data sets (\cref{subsubsec:synthetic}), we run experiments with ONOS~\cite{ONOS,ONOSDeployments}, an open SDN platform used by sizeable operator networks around the globe~\cite{ONOS,ONOSDeployments}.

\begin{figure}[b]
\centering
\includegraphics[scale=.89]{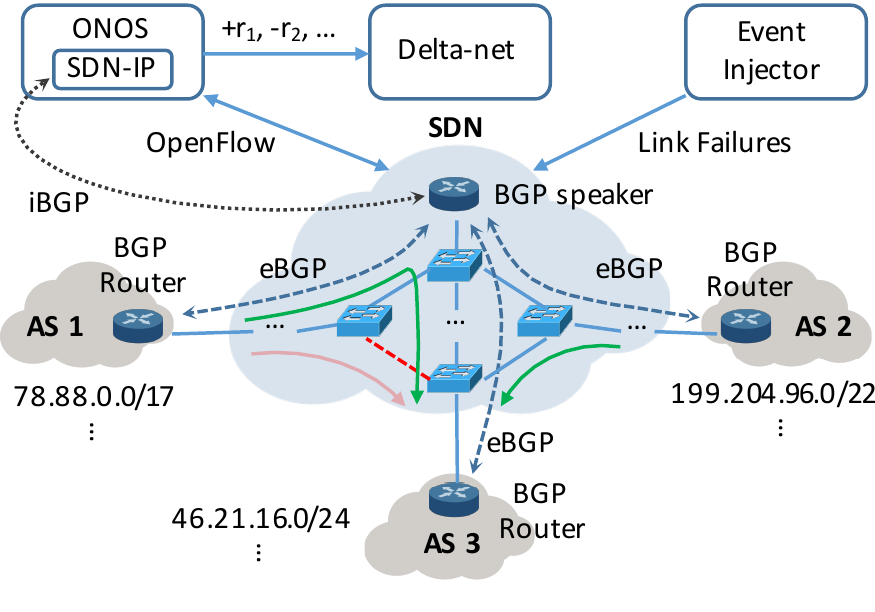}
\caption{Experimental setup with SDN-IP application.}
\label{fig:sdn-ip}
\end{figure}

To obtain a relevant and realistic experimental setup, we run SDN-IP~\cite{LHKMKAWB2013,SDNIP}, an important ONOS application that allows an ONOS-controlled network to interoperate with external autonomous networks (ASes). This interoperability is achieved as follows (\Cref{fig:sdn-ip}). Inside the ONOS-controlled network reside \emph{Border Gateway Protocol} (BGP) speakers (in our experimental setup there is exactly one internal BGP speaker) that use eBGP to exchange BGP routing information with the border routers of adjacent external ASes. This information, in turn, is propagated inside the ONOS-controlled network via iBGP. As sketched in the upper half of~\Cref{fig:sdn-ip}, SDN-IP listens to these iBGP messages and requests ONOS to dynamically install IP forwarding rules such that packets destined to an external AS arrive at the correct border router. In doing so, SDN-IP sets the priority of rules according to the longest prefix match where rules with longer prefix lengths receive higher priority. For each rule insertion and removal (depicted by $+r_1$ and $-r_2$ in~\Cref{fig:sdn-ip}), Delta-net checks the resulting data plane.

For our experiments, we run SDN-IP in a single ONOS instance. We use Mininet~\cite{LHMcK2010} to emulate a network of sixteen Open vSwitches~\cite{PPKJZRGWS2015}, configured according to the Airtel network topology (AS~9498)~\cite{KNFBR2011}. We connect each of these OpenFlow-compliant switches~\cite{OpenFlow} to an external border router that we emulate using Quagga~\cite{Quagga}. We configure Quagga such that each border router advertises one hundred IP prefixes, which we randomly select from over half a million real-world IP prefixes gathered from the Route Views project~\cite{RouteViews}, resulting in a total of $1,600$ unique (but possibly overlapping) IP prefixes.

Our experiments in~\cref{subsubsec:rule-insertions-and-removals} exploit the fact that SDN-IP relies on ONOS to reconfigure the OpenFlow switches when parts of the network fail. Since network failures happen frequently~\cite{BK2014} and pose significant challenges for real-time data plane checkers~\cite{KCZVMcKW2013,KZZCG2013}, we can generate interesting data sets by systemically failing links, controlled by the `Event Injector' process in the upper right half of~\Cref{fig:sdn-ip}. In particular, the Airtel~1 data set contains the rule insertions and removals triggered by failing a single inter-switch link at a time, recovering each link before failing the next one. Such a link failure (dashed red edge) is illustrated in the left half of~\Cref{fig:sdn-ip}, causing ONOS to reconfigure the data plane so that a new path is established (green arrow on the left) that avoids the failed link, which caused disruption to earlier network traffic (red arrow). In the case of Airtel~2, we automatically induce all $2$-pair link failures (separately failing the first link and then the second one), including their recovery.

We also wanted to study a larger number of rules and IP prefixes, but were limited due to technical issues with ONOS. We worked around these limitations by using a 4-switch ring network. In this smaller ring topology, we configure each Quagga instance to advertise $5,000$ IP prefixes (rather than only $100$ IP prefixes as in the Airtel experiments), again randomly selected from the Route Views project~\cite{RouteViews}. We do not fail any links. Instead, we only collect the rules generated by SDN-IP, a process we repeat fourteen times with different IP prefixes. This workaround yields the 4Switch data set in~\Cref{table:data-sets}, comprising $1.12$ million rules. In contrast to the previously described data sets, all of the operations in the 4Switch data set are rule insertions.

\subsection{Experimental results}
\label{subsec:experimental-results}

Our experiments separately measure Delta-net's performance in checking individual rule updates (\cref{subsubsec:rule-insertions-and-removals}) and handling a ``what if'' scenario (\cref{subsec:checking-link-failures}). In both cases, at the cost of higher memory usage, Delta-net is more than $10\times$ faster than the state-of-the-art. We run our experiments on an Intel Xeon CPU with 3.47\,GHz and 94\,GB of RAM. Since our implementation is single-threaded (\cref{subsec:implementation}), we utilize only one out of the 12 available cores.

\subsubsection{Checking network updates}
\label{subsubsec:rule-insertions-and-removals}

To evaluate Delta-net's performance with respect to rule insertions and removals, we build the delta-graph (\cref{subsec:revisited-design-goals}) for each operation, and find in it all forwarding loops, a common network-wide invariant~\cite{KVM2012,KCZVMcKW2013,YL2013,KZZCG2013,ZZYJJLMV2014}. We process the rules in each data set in the order in which they appear in the data sets (\cref{subsec:data-sets}).

\begin{table*}
\centering
\scalebox{0.76}{
\begin{tabular}{l||r|r|r|r|r|r|r|r}
  \toprule
  & \textbf{Berkeley} & \textbf{INET}  & \textbf{RF~1755} & \textbf{RF~3257} & \textbf{RF~6461} &  \textbf{Airtel~1} & \textbf{Airtel~2} & \textbf{4Switch} \\
    \midrule
  Total number of atoms
  & $668,520$ & $563,480$ & $726,535$ & $726,535$ & $726,535$ & $2,799$ & $2,799$ & $443,443$  \\  
 
  Median rule processing time
  & $4 \,\mu s$ & $5 \,\mu s$ & $4\,\mu s$ & $5\,\mu s$ & $5\,\mu s$  &  $2\,\mu s$ & $1\,\mu s$ & $4\,\mu s$  \\
  
  Average rule processing time
  & $5 \,\mu s$ & $41 \,\mu s$ & $11\,\mu s$ & $22\,\mu s$ & $20\,\mu s$ & $3\,\mu s$ & $3\,\mu s$ & $5\,\mu s$   \\
  
  Percentage $< 250\, \mu s$ & $99.9\%$ & $98.5\%$ & $99.8\%$ & $99.6\%$ & $99.7\%$ & $99.9\%$ & $99.9\%$ & $99.9\%$ \\
  \bottomrule
\end{tabular}
}
\caption{Experimental results using Delta-net, measuring rule insertions and removals.}
\label{table:experimental-results}
\end{table*}

\begin{figure}
\centering
\includegraphics{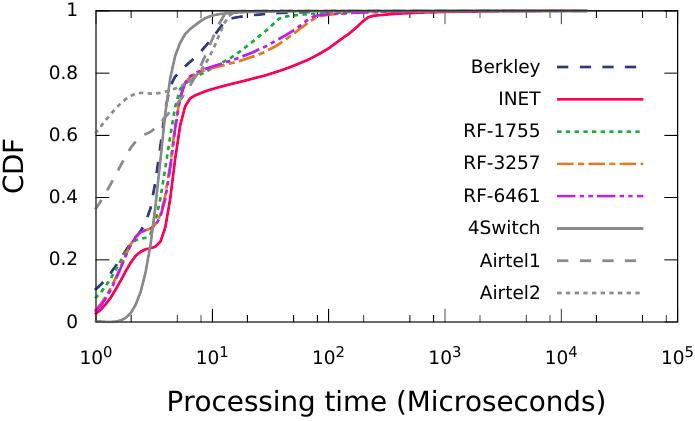}
\caption{CDF of combined time (microseconds) for processing a rule update and checking for forwarding loops.}
\label{fig:cdf}
\end{figure}

\Cref{table:experimental-results} summarizes our experimental results for measuring the checking of rule insertions and removals. The first row in~\Cref{table:experimental-results} shows that the total number of atoms is much smaller than the total number of rules in the network (recall \Cref{table:data-sets}), suggesting a significant degree of commonality among IP prefix rules that atoms effectively exploit. Furthermore, for all data sets, the median and average rule processing time is less than $5$ and $41$ microseconds, respectively, which includes the checking of forwarding loops. On closer inspection, as shown in the last row of~\Cref{table:data-sets}, Delta-net processes rule updates and checks for the existence of forwarding loops in less than $250$ microseconds for at least $98.5\%$ of cases. The combined time for processing a rule update and finding \emph{all} forwarding loops in the corresponding delta-graph (\cref{subsec:revisited-design-goals}) is visualized by the cumulative density function (CDF) in~\Cref{fig:cdf}. It shows that the INET data set~\cite{ZZYJJLMV2014} (solid red line) is one of the more difficult ones for Delta-net. We remark that Delta-net's memory usage never exceeds the available memory on our machine (Appendix~\ref{appendix:memory-usage}).

Our measurements are significant because earlier experiments with Veriflow~\cite{KZZCG2013} result in an average verification time of $380$ microseconds, whereas Delta-net verifies rule insertions and removals in often tens of microseconds, and $41$ microseconds on average even on the largest data set, INET. This comparison is meaningful because our data sets are significantly larger than previous ones~\cite{KZZCG2013,C2014,KCZVMcKW2013,YL2013}. Moreover, two of our data sets (Airtel~1~and~2) are derived from a real-world software-defined networking application while causing an extensive number of link failures in the network, which were previously shown to lead to longer verification times~\cite{KCZVMcKW2013,KZZCG2013}. Our experiments therefore provide strong evidence that Delta-net can be at least one order of magnitude faster compared to Veriflow~\cite{KZZCG2013}. Since neither Veriflow's implementation (or its algorithm) nor any of the data sets used for its experimental evaluation are publicly available, and neither its time nor space complexity is specified, we further quantify the differences between Delta-net and Veriflow by re-implementing a consistent interpretation of Veriflow, as described next.

Our re-implementation of Veriflow, which we call \emph{Veriflow-RI}, is not intended to be a full-feature copy of Veriflow, but rather a re-implementation of their core idea to enable an honest comparison with Delta-net. Specifically, Veriflow-RI is designed for matches against a single packet header field. This explains why Veriflow-RI uses a one-dimensional trie data structure in which every node has at most two children (rather than three~\cite{KZZCG2013}). We optimize the computation of equivalence classes and construction of forwarding graphs. Note that these optimizations may not be possible in the original Veriflow implementation with its ternary trie data structure, and Veriflow-RI may therefore be faster than Veriflow~\cite{KZZCG2013}. We remark that Veriflow-RI's space complexity is linear in the number of rules in the network, whereas its time complexity is quadratic, rather than quasi-linear as in the case of Delta-net (\Cref{theorem:complexity-analysis}).


While Delta-net is only approximately $4\times$ faster than Veriflow-RI on the Airtel data set, on the INET data set, Delta-net is approximately $6\times$ faster than Veriflow-RI. This gap widens on the RF~3257~and~6461 data sets where Delta-net is approximately $7\times$ faster than Veriflow-RI. In turn, however, Veriflow-RI consumes $5-7\times$ less memory than Delta-net (Appendix~\cref{appendix:memory-usage}).

It is therefore natural to ask whether this trade-off in space and time is worth it. Next, we answer this question affirmatively by showing that Delta-net can check properties for which Veriflow often times out. This difference in run-time performance is due to the fact that Delta-net incrementally maintains flow information of every packet in the entire network, whereas Veriflow recomputes the forwarding graph for each affected equivalence class. What is remarkable is that Delta-net achieves this extra bookkeeping without limiting the checking of individual network updates (see previous paragraph).

\vspace{-0.5em}
\subsubsection{Beyond network updates}
\label{subsec:checking-link-failures}

We show how Delta-net can go beyond traditional data plane checks per network update. To do so, we consider the following question, which was previously posed by~\cite{KZZCG2013}, as an exemplar of a ``what if'' query: What is the fate of packets that are using a link that fails? We interpret their question to mean that Veriflow has to construct forwarding graphs for all packet equivalence classes that are affected by a link failure. This is known to be a difficult task for Veriflow since it requires the construction of at least a hundredfold more forwarding graphs compared to checking a rule insertion or removal (\cref{subsubsec:rule-insertions-and-removals}). Here, our experiment quantifies how much Delta-net gains by incrementally transforming a single-edge labelled graph instead of constructing multiple forwarding graphs.

For our experiments, we generate a consistent data plane from all the rule insertions in the five synthetic and 4Switch data sets in~\Cref{table:data-sets}, respectively. And in the case of Airtel, we extract a consistent data plane snapshot from ONOS. The total number of resulting rules in each data plane is shown in the second column of~\Cref{table:delta-net-vs-veriflow-ri}. For all of these seven data planes, we answer which packets and parts of the network are affected by a hypothetical link failure. The verification task therefore is to represent via one or multiple graphs all flows of packets through the network that would be affected when a link fails. The third column in~\Cref{table:data-sets} (number of links) corresponds to the number of queries we pose, except for the new Airtel data plane snapshot where we pose 158 queries.

Since Delta-net already maintains network-wide packet flow information, we expect it to perform better than Veriflow-RI.\footnote{Recall from previous experiments (\cref{subsubsec:rule-insertions-and-removals}), Delta-net's extra bookkeeping poses no performance problems for checking network updates.} The third and fourth column in~\Cref{table:delta-net-vs-veriflow-ri} quantify this performance gain by showing the average query time of Veriflow-RI and Delta-net, respectively.  On three data planes, Veriflow-RI exceeds the total run-time limit of 24 hours, whereas the longest running Delta-net experiment takes a total of $3.2$ hours. When these time outs in Veriflow-RI occur, we report its incomplete average query time $t$ as `$t^{\dagger}$'. We find that Delta-net is usually more than $10\times$ faster than Veriflow-RI (even if Delta-net checks for forwarding loops, as reported in the last column). Since Delta-net is very fast in maintaining the flow of packets, the difference between the last two columns in~\Cref{table:delta-net-vs-veriflow-ri} shows that Delta-net's processing time is dominated by the property check (here, forwarding loops). In contrast to Delta-net, Veriflow's processing time is reportedly dominated by the construction of forwarding graphs~\cite{KZZCG2013}.


\begin{table}[b]
\centering
\scalebox{0.76}{
\begin{tabular}{l|r|r|r|r}
  \toprule
 \multicolumn{1}{c|}{\textbf{Data plane}}                     & \multicolumn{1}{|c|}{\textbf{Rules}} & \multicolumn{3}{|c}{\textbf{Average query time ($ms$)}} \\
                      &                & \small Veriflow-RI & \small Delta-net & \small $+\,$Loops \\
    \midrule
    Berkeley          & $12,817,902$  & $3,073.0$     & $\mathbf{4.7}$  & $93.3$ \\
    INET              & $124,733,556$ & $29,117.5^{\dagger}$      & $\mathbf{0.7}$  & $2,888.6$ \\
    RF~1755           & $33,732,869$  & $8,100.6$     & $\mathbf{1.3}$  & $897.4$ \\
    RF~3257           & $74,492,920$  & $17,645.3^{\dagger}$      & $\mathbf{1.0}$  & $2.6$ \\
    RF~6461           & $75,005,738$  & $17,594.5^{\dagger}$      & $\mathbf{0.4}$  & $0.4$ \\
    Airtel            & $38,100$      & $4.5$        & $\mathbf{0.04}$ & $2.3$ \\
    4Switch           & $1,120,000$   & $433.4$      & $\mathbf{21.1}$ & $128.1$ \\
 \bottomrule
\end{tabular}
}
\caption{Experimental results for ``what if'' link failures.\vspace{-0.4em}}
\label{table:delta-net-vs-veriflow-ri}
\end{table}

\section{Related work}
\label{sec:related-work}

In this section, we discuss related works in the literature.


\vspace{-1em}
\paragraph{Stateful networks.} One of the earliest stateful network analysis techniques~\cite{CVPDR2012} proposes symbolic execution of OpenFlow applications using a simplified model of OpenFlow network switches. VeriCon~\cite{BBGIKSSV2014} uses an SMT solver to automatically prove the correctness of simple SDN controllers. FlowTest~\cite{FS2014} investigates relevant AI planning techniques. SymNet~\cite{SPNR2016} symbolically analyzes stateful middleboxes through additional fields in the packet header. Unlike~\cite{CVPDR2012}, BUZZ~\cite{FYTCS2016} adopts a symbolic model-based testing strategy~\cite{UPL2012} as a way to capture the state of forwarding devices. Most recent complexity results~\cite{VAPRSSS2016} are the first step towards a taxonomy of decision procedures in this research area. Real-time network verification techniques (see next paragraph) can be extended to check safety properties that depend on the state of the SDN controller~\cite{BZZMRW2014}.

\vspace{-1em}
\paragraph{Stateless networks.} The seminal work of Xie et al.~\cite{XZMZGHR2005} introduces stateless data plane checking to which Delta-net belongs. The research that emerged from~\cite{XZMZGHR2005} can be broadly divided into offline~\cite{YMSCCM2006,AMEE2009,JS2009,NBDFK2010,AA2010,MKACGK2011,SSYPG2012,KVM2012,MCD2015,FFPWGMM2015,LBGJV2015} and online~\cite{KZZCG2013,KCZVMcKW2013,YL2013} approaches. The offline approaches encode the problem into Datalog~\cite{FFPWGMM2015,LBGJV2015} or logic formulas that can be checked for satisfiability by constructing a Binary Decision Diagram~\cite{YMSCCM2006,AMEE2009} or calling an SAT/SMT solver~\cite{JS2009,NBDFK2010,AA2010,MKACGK2011,SSYPG2012,JBOK2014,MCD2015}. By contrast, all modern online approaches~\cite{KZZCG2013,KCZVMcKW2013,YL2013} partition in some way the set of all network packets. In particular, the partitioning scheme described in~\cite{KVM2012}, on which~\cite{KZZCG2013} is based, dynamically computes equivalence classes by propagating ternary strings in the network, whereas more recent work~\cite{KCZVMcKW2013,YL2013,NGRSG2016}, including ours, pre-compute network packet partitions prior to checking a verification condition. Our work could be used in conjunction with network symmetry reduction techniques~\cite{PBLRV2016}. Custom network abstractions can be very useful for restricted cases~\cite{GJVAM2016}. While potentially less efficient, our work is more general than~\cite{GJVAM2016}, and most closely related to~\cite{KZZCG2013,C2014,KCZVMcKW2013,YL2013,ZZYJJLMV2014,NGRSG2016}, which we discuss in turn. The complexity of the most prominent of these works, including Veriflow~\cite{KZZCG2013} and NetPlumber~\cite{KCZVMcKW2013}, is summarized in work~\cite[Section II]{L2016} that is independent from ours.

Veriflow~\cite{KZZCG2013} constructs multiple forwarding graphs that may significantly overlap (\cref{subsec:example}). Our algorithm exploits this overlapping and transforms a single edge-labelled graph instead. Moreover, Veriflow relies on the fact that overlapping IP prefixes can be efficiently found using a trie data structure~\cite{KZZCG2013}. By contrast, atoms are generally not expressible as a single IP prefix. For example, atom $[0 : 10)$ in~\Cref{fig:half-closed-intervals} can only be represented by the union of at least two IP prefixes.

Chen~\cite{C2014} shows how to optimize Veriflow~\cite{KZZCG2013}, while retaining its core algorithm. Similar to~\cite{C2014}, we represent IP prefixes in a balanced binary search tree. Unlike~\cite{C2014}, however, our representation serves as a built-in index of half-closed intervals through which we address fundamental limitations of Veriflow (\cref{subsec:example}).

NetPlumber~\cite{KCZVMcKW2013} incrementally creates a graph that, in the worst case, consists of $R^2$ edges where $R$ is the number of rules in the network. In contrast to NetPlumber, Delta-net maintains a graph whose size is proportional to the number of links in the network, which is usually much smaller than $R$. Since the number of atoms tends to be much less than $R$ (\cref{sec:experiments}), Delta-net has an asymptotically smaller memory footprint than NetPlumber.

Yang and Lam~\cite{YL2013} propose a more compact representation of forwarding graphs that reduces the task of data plane checking to intersecting sets of integers. For the restricted, but common, case of checking IP forwarding rules, our algorithm is asymptotically faster than theirs. Our algorithm, however, does not find the unique minimal number of packet equivalence classes, cf.~\cite{YL2013}.

More recent work for stateless and non-mutating data plane verification~\cite{NGRSG2016} encodes a canonical form of ternary bit-vectors, and shows on small data sets with a few thousand rules that their encoding performs better than Yang and Lam~\cite{YL2013}'s algorithm. It would be interesting to repeat these experiments on our, significantly larger, data sets.

Finally, Libra~\cite{ZZYJJLMV2014} may be used for incrementally checking network updates, but it requires an in-memory ``streaming'' MapReduce run-time, whereas Delta-net avoids the overheads of such a distributed system. Since Libra's partitioning scheme into disjoint subnets is orthogonal to our algorithm, however, it would be interesting to leverage both ideas together in future work.

\section{Concluding remarks}

In this paper, we presented Delta-net (\cref{sec:delta-net}), a new data plane checker that is inspired by program analysis techniques in the sense that it automatically refines a lattice-theoretical abstract domain to precisely represent the flows of all packets in the entire network. We showed that this matters from a theoretical and practical point of view: Delta-net is asymptotically faster and/or more space efficient than prior work~\cite{KZZCG2013,KCZVMcKW2013,YL2013}, and its new design facilitates Datalog-style use cases~\cite{FFPWGMM2015,LBGJV2015} for which the transitive closure of many or all packet flows needs to be efficiently computed (\cref{subsec:revisited-design-goals}). In addition,  Delta-net can be used to analyze catastrophic network events, such as link failures, for which current incremental techniques are less effective. To show this experimentally (\cref{sec:experiments}), we ran an adaptation of the link failure experiments by Khurshid et al.~\cite{KZZCG2013} on data sets that are significantly larger than previous ones. For this exemplar ``what if'' scenario, we found that Delta-net is several orders of magnitude faster than the state-of-the-art (\Cref{table:delta-net-vs-veriflow-ri}). Our work therefore opens up interesting new research directions, including testing scenarios under different combinations of failures, which have been shown to be effective for distributed systems, e.g.~\cite{YLZRZZJS2014}.

\vspace{-1em}
\paragraph{Future work.} One advantage of Delta-net is that its main loops over atoms in~\Cref{alg:insert-rule}~and~\ref{alg:remove-rule} are highly parallelizable. In addition, (stateless) packet modification of IP prefixes can be easily supported without substantial changes to the data structures by augmenting the edge-labelled graph with the necessary information on how atoms are transformed along hops. We are also studying an improved version of Delta-net that avoids the quadratic space complexity by exploiting properties of IP prefixes. Finally, since a naive implementation of Delta-net is exponential in the number of range-based packet header fields (as is Veriflow's~\cite[Section II]{L2016}), it would be interesting to guide further developments into multi-range support in higher dimensions using the `overlapping degree' among rules~\cite{L2016}.

\small
\paragraph{Acknowledgements.} We would like to thank Sho Shimizu, Pingping Lin and members of the ONOS developer mailing list for technical support. We thank Rao Palacharla, Nate Foster and Mina Tahmasbi for their invaluable feedback on an early draft of this paper. We also would like to thank Ratul Mahajan and the anonymous reviewers of NSDI for their detailed comments and helpful suggestions.

{\footnotesize \bibliographystyle{acm}
\bibliography{delta-net}}

\begin{appendices}
\section{Illustration of Boolean lattice}
\label{appendix:Hasse-diagram}

Delta-net is based on ideas from lattice theory.\footnote{For interested readers, a good introduction to lattice theory, whose applications in computer science are pervasive, can be found in~\cite{DP2002}} In particular, Delta-net leverages the concept of atoms, a form of mutually disjoint ranges that make it possible to analyze all Boolean combinations of IP prefix forwarding rules in a network. The fact that atoms induce a Boolean lattice is illustrated by the Hasse diagram~\cite{DP2002} in~\Cref{fig:Boolean-lattice} where atoms (depicted in bold) correspond to $\atom{0}$, $\atom{1}$ and $\atom{2}$ in~\Cref{fig:half-closed-intervals}, respectively. 

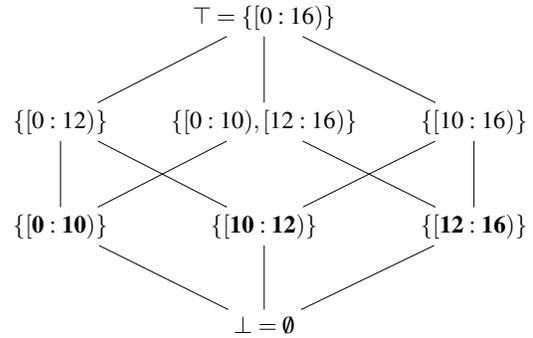
\begin{figure}
\centering
\scalebox{0.92}{
\xymatrix@C=0.7em@R=2.7em{
&&&\top = \{[0:16)\}&&\\
&\{[0:12)\}\ar@{-}[rru]&&\{[0:10),[12:16)\}\ar@{-}[u]&&\{[10:16)\}\ar@{-}[llu]\\
&\{\mathbf{[0:10)}\}\ar@{-}[u]\ar@{-}[rru]&&\{\mathbf{[10:12)}\}\ar@{-}[llu]\ar@{-}[rru]&&\{\mathbf{[12:16)}\}\ar@{-}[u]\ar@{-}[llu] \\
&&&\bot = \emptyset\ar@{-}[llu]\ar@{-}[u]\ar@{-}[rru]&&
}
}
\caption{Boolean lattice induced by the atoms (bold) in~\Cref{fig:half-closed-intervals}, assuming 4-bit numbers for simplicity.}
\label{fig:Boolean-lattice}
\end{figure}

\section{Proof of complexity analysis}
\label{appendix:proof-of-complexity-analysis}

In this appendix, we sketch the proof of the asymptotic worst-case time complexity of~\Cref{alg:insert-rule}~and~\ref{alg:remove-rule}.

\begin{proof}[Proof of~\Cref{theorem:complexity-analysis}]
We analyze \textproc{Insert\_Rule}. Each atom split (\cref{line:insert-rule-split-atoms-begin}-\ref{line:insert-rule-split-atoms-end}) requires copying the owner information from an existing atom to a newly created atom. For insertion of $R$ rules, resulting in $K$ atoms, this requires $O(RK)$ steps in the worst-case. In each insertion, the adjustment of labels and retrieval of the balanced binary search tree (BST) (\cref{line:insert-rule-get-bst}) are amortized constant-time operations per atom.  Inserting each rule into the BST and finding the highest-priority rule per atom (\cref{line:insert-rule-get-highest-priority-rule}) are $O(\log M)$. By the loop (\cref{line:insert-rule-atom-loop-begin}-\ref{line:insert-rule-atom-loop-end}), we get $O(RK + RK \log M) = O(RK \log M)$, concluding the proof. A similar argument proves the claim for \textproc{Remove\_Rule}.
\end{proof}

\section{Comparison to previous data sets}
\label{appendix:data-sets}

In this appendix, we discuss how our data sets compare to previous ones used in the experimental evaluation of Veriflow~\cite{KZZCG2013}.

In particular, it is natural to ask how our RF~1755 data set in~\Cref{table:data-sets} compares to the one used in a previous Veriflow experiment~\cite{KZZCG2013}, which was constructed from 5 million BGP RIP entries and by `replaying' 90,000 BGP updates. While the resulting total number of IP prefix rules in the original RF~1755 data set is not reported, the authors of the Veriflow paper note that ``[t]he largest number of ECs (equivalence classes) affected by a single rule was 574; the largest verification latency was
$159.2\,ms$ due to an update affecting 511 ECs.'' For our experiments, we expect this number to be different, since we had to generate a new data set.

Running Veriflow-RI (\cref{subsubsec:rule-insertions-and-removals}) on our RF~1755 data set, we find that the maximum number of affected ECs on rule insertions is $319,681$, which is significantly larger than the original experimental evaluation of Veriflow~\cite{KZZCG2013}.

\section{Memory usage}
\label{appendix:memory-usage}

In this appendix, we report the detailed memory consumption of Delta-net (\cref{sec:delta-net}) and Veriflow-RI (\cref{subsubsec:rule-insertions-and-removals}) using our eight data sets (\cref{subsec:data-sets}, see~\Cref{table:data-sets}).

\Cref{table:memory-usage} quantifies the memory usage of Delta-net and Veriflow-RI. In all cases, Delta-net consumes between $5$ and $7$ times more space than Veriflow-RI. This increase in memory consumption is offset, however, by the fact that Delta-net keeps track of the forwarding behaviour of all packets, and as a result can check properties that Veriflow-RI cannot. Nevertheless, as discussed for future work (\cref{sec:related-work}), we are actively working on asymptotically reducing the memory consumption of Delta-net.

\begin{table}
\centering
\scalebox{0.76}{
\begin{tabular}{l|r|r}
  \toprule
   \multicolumn{1}{c|}{\textbf{Data set}} & \multicolumn{2}{c}{\textbf{Memory usage} (\textit{MB})} \\
    & \multicolumn{1}{c|}{\small Veriflow-RI} & \multicolumn{1}{c}{\small Delta-net} \\
    \midrule
    Berkeley          & $1,089$ & $6,208$ \\
    INET              & $9,776$ & $63,563$ \\
    RF~1755           & $2,713$ & $16,937$ \\
    RF~3257           & $5,882$ & $40,716$ \\
    RF~6461           & $5,920$ & $39,481$ \\
    Airtel~1          & $7$    & $61$ \\
    Airtel~2          & $9$    & $74$ \\
    4Switch           & $154$  & $785$ \\
  \bottomrule
\end{tabular}
}
\caption{Memory usage of Delta-net and Veriflow-RI.}
\label{table:memory-usage}
\end{table}




\end{appendices}

\end{document}